\documentclass[12pt]{amsart}

\usepackage{amsthm}
\usepackage{amssymb}
\usepackage{amsmath}

\numberwithin{equation}{section}
\newtheorem{thm}{Theorem}[section]
\newtheorem{prop}[thm]{Proposition}
\newtheorem{lem}[thm]{Lemma}

\theoremstyle{definition}
\newtheorem{definition}[thm]{Definition}

\begin{document}

\title[Discrete delta Bose gas and affine Hecke algebra]
{A discrete analogue of periodic delta Bose gas and affine Hecke algebra}
\author{Yoshihiro Takeyama}
\address{Division of Mathematics, 
Faculty of Pure and Applied Sciences, 
University of Tsukuba, Tsukuba, Ibaraki 305-8571, Japan}
\email{takeyama@math.tsukuba.ac.jp}

\begin{abstract}
We consider an eigenvalue problem for a discrete analogue of the Hamiltonian 
of the non-ideal Bose gas with delta-potentials on a circle. 
It is a two-parameter deformation of the discrete Hamiltonian for 
joint moments of the partition function of the O'Connell-Yor semi-discrete polymer. 
We construct the propagation operator by using integral-reflection operators, 
which give a representation of the affine Hecke algebra. 
We also construct eigenfunctions by means of the Bethe ansatz method. 
In the case where one parameter of our Hamiltonian is equal to zero,   
the eigenfunctions are given by specializations of the Hall-Littlewood polynomials.  
\end{abstract}
\maketitle

{\small
\noindent{\it 2010 Math. Subj. Class.} 
39A12, 81R12}
\medskip 

\setcounter{section}{0}
\setcounter{equation}{0}


\section{Introduction}

In this paper we study a discrete analogue of the Hamiltonian of the non-ideal Bose gas with 
delta-potential interactions, which we call the delta Bose gas for short. 
The eigenvalue problem for the delta Bose gas with periodic boundary condition 
was solved by Lieb and Liniger by means of the Bethe ansatz method \cite{LL}. 
The Hamiltonian of the system is given by 
\begin{align}\label{eq:delta-bose}
{}-\Delta+\sum_{1 \le i<j \le k \atop m \in \mathbb{Z}}\delta(x_{i}-x_{j}+m).  
\end{align}
The potential is supported by affine hyperplanes associated 
with the affine root system of type $A_{k-1}^{(1)}$. 
Gutkin and Sutherland generalized the periodic delta Bose gas for all affine root systems \cite{GS}. 

In \cite{HO} Heckman and Opdam studied the root system generalization 
of the delta Bose gas on a line, 
and revealed a connection with harmonic analysis on 
homogeneous spaces of semisimple groups.  
A key observation is that the propagation operators give 
a representation of the degenerate (or graded) Hecke algebra \cite{D, L}. 
For the periodic case, Emsiz, Opdam and Stokman \cite{EOS} 
found that the underlying symmetry is governed by 
Cherednik's degenerate double affine Hecke algebra \cite{C}. 

In this paper we consider a discrete analogue of the Hamiltonian \eqref{eq:delta-bose}. 
An integrable discretization of the delta Bose-gas associated with 
root systems has been already proposed and studied by van Diejen \cite{vD, vD2}. 
It consists of discrete Laplace operators and 
boundary conditions for functions on a fundamental domain 
of the affine Weyl group in the weight lattice. 
The associated eigenvalue problem can be solved by means of 
Macdonald's spherical functions. 

The discrete version which we will consider is different from van Diejen's model 
and has an origin in the study of integrable stochastic models. 
The Kardar-Parisi-Zhang (KPZ) equation is a stochastic partial differential equation 
for height function $\mathcal{H}$ of growing interfaces (see the review \cite{Cor} for details). 
The Cole-Hopf solution $\mathcal{Z}:=\exp{(-\mathcal{H})}$ satisfies the stochastic heat equation. 
The fact is that the $n$-th moment of $\mathcal{Z}$ satisfies an evolution equation 
with the Hamiltonian of the delta Bose gas on a line with $n$ particles. 
An integrable discretization of the KPZ equation is 
the $q$-deformed totally asymmetric simple exclusion process ($q$-TASEP) 
(see \cite{BC}, Section 3.3.2). 
In a scaling limit $q$-TASEP goes to the O'Connell-Yor semi-discrete directed polymer \cite{OY}. 
The joint moment 
$\tilde{v}(\tau; \overrightarrow{n}) \, 
(\tau \in \mathbb{R}_{>0}, \overrightarrow{n} \in (\mathbb{Z}_{\ge 0})^{k})$ of its partition function 
satisfies the following evolution equation \cite{BCS}: 
\begin{align}\label{eq:OC-Y}
\frac{d}{d\tau}\tilde{v}(\tau; \overrightarrow{n})=\tilde{H}\tilde{v}(\tau; \overrightarrow{n}), 
\quad 
\tilde{H}:=\sum_{i=1}^{k}\nabla_{i}+\sum_{1 \le i<j \le k}\delta_{n_{i}, n_{j}},  
\end{align}
where $\nabla_{i}$ is the difference operator 
\begin{align*}
(\nabla_{i} f)(\overrightarrow{n}):=
f(n_{1}, \ldots , n_{i}-1, \ldots , n_{k})-f(n_{1}, \ldots , n_{i}, \ldots , n_{k}). 
\end{align*}

In this paper we consider a two-parameter deformation of 
the Hamiltonian $\tilde{H}$ with periodic boundary condition. 
It acts on the space of $\mathbb{C}$-valued functions on the $k$-dimensional orthogonal lattice 
$X=\oplus_{i=1}^{k}\mathbb{Z}v_{i}$ 
and is given by 
\begin{align*}
H:=\sum_{i=1}^{k}\beta^{d_{i}^{-}}(t_{v_{i}}-\alpha d_{i}^{+}),   
\end{align*}
where $\alpha \in \mathbb{C}$ and $\beta \in \mathbb{C}^{\times}$ are parameters, 
and $t_{v_{i}}$ is the shift operator $(t_{v_{i}}f)(x):=f(x-v_{i})$. 
The functions $d_{i}^{\pm}$ count positive roots of type $A_{k-1}^{(1)}$ 
whose values at $x$ are non-positive multiple of the system size 
(see \eqref{eq:def-d1} and \eqref{eq:def-d2} below). 
Setting $\beta=1$ and $\alpha=-1$ we recover $\tilde{H}$ 
with periodic boundary condition up to an additive constant.   

The main result of this paper is construction of the propagation operator $G$ 
which sends an eigenfunction of ``half Laplacian'' $\sum_{i=1}^{k}t_{v_{i}}$ 
to that of the Hamiltonian $H$ with the same eigenvalue (see Theorem \ref{thm:main} below). 
To define $G$ we make use of a discrete analogue of the integral-reflection operators 
due to Yang \cite{Y} for the case of type $A$ and Gutkin \cite{G} for the general case. 
Van Diejen and Emsiz \cite{ED} constructed the discrete version from 
a polynomial representation of the affine Hecke algebra. 
We follow their construction  but start from more general divided difference operators 
satisfying the braid relations, 
which are classified by Lascoux and Sch\"utzenberger \cite{LS}. 
Then our integral-reflection operators also give a representation of 
the affine Hecke algebra of type $GL_{k}$. 

The propagation operator enables us to 
construct symmetric and periodic eigenfunctions 
for the Hamiltonian $H$ by means of the Bethe ansatz method, 
which we call Bethe wave functions.   
In the case of $\alpha=0$ the Bethe wave functions 
are given by specializations of the Hall-Littlewood polynomials \cite{Mac} 
(see \eqref{eq:HL} below). 

The paper is organized as follows. 
In Section \ref{sec:pre} we prepare some notation and lemmas about 
the affine root system of type $A_{k-1}^{(1)}$. 
We define the operator $H$ in Section \ref{sec:hamiltonian} 
and the associated integral-reflection operators in Section \ref{sec:refl}. 
In Section \ref{sec:main} we prove the main theorem. 
The Bethe wave functions are constructed in Section \ref{sec:bethe}.


\section{Preliminaries}\label{sec:pre}

Throughout this paper we fix two integers $k \ge 2$ and $L \ge 1$. 
Let $V$ be the $k$-dimensional Euclidean space with inner product $\langle \cdot , \cdot \rangle$, 
and $V^{*}$ the linear dual of $V$. 
We also write $\langle \cdot , \cdot \rangle$ for the associated inner product on $V^{*}$. 
Fix an orthogonal basis $\{v_{i}\}_{i=1}^{k}$ of $V$, and set 
\begin{align*}
X:=\oplus_{i=1}^{k}\mathbb{Z}v_{i}.  
\end{align*}
We denote by $\{\epsilon_{i}\}_{i=1}^{k}$ the dual basis corresponding to $\{v_{i}\}_{i=1}^{k}$. 

For $\xi \in V^{*}\setminus\{0\}$ we define the co-vector $\xi^{\vee} \in V$ by 
the property 
\begin{align*}
\eta(\xi^{\vee})=2\frac{\langle \eta, \xi \rangle}{\langle \xi, \xi \rangle} \qquad 
(\forall{\eta} \in V^{*}).  
\end{align*}

Let ${\rm Aff}(V):=V^{*} \oplus \mathbb{R}\delta$ be 
the space of affine linear functions on $V$, where 
$\delta(v)=1$ for all $v \in V$. 
Denote the gradient map by $D \, : {\rm Aff}(V) \twoheadrightarrow V^{*}$. 

For $\phi \in {\rm Aff}(V)$, the orthogonal reflection 
$s_{\phi} \, : V \to V$ with respect to 
the affine hyperplane $V_{\phi}:=\{v \in V | \phi(v)=0\}$ is given by  
\begin{align*}
s_{\phi}(v):=v-\phi(v)(D\phi)^{\vee}.  
\end{align*}
Define the translation map $t_{v'} \, : V \to V$ for $v' \in V$ by 
\begin{align*}
t_{v'}(v):=v+v'.  
\end{align*}
We also denote $s_{\phi}$ and $t_{v}$ for the corresponding transpositions 
acting on the space of functions on $V$, that is, 
$(s_{\phi}f)(v):=f(s_{\phi}(v))$ and $(t_{v'}f)(v):=f(t_{-v'}(v))$. 

Set $\alpha_{ij}:=\epsilon_{i}-\epsilon_{j}$ for $1\le i, j \le k$. 
The subset $R_{0}:=\{\alpha_{ij} \, | \, 1 \le i, j \le k, \, i\not=j\}$ 
of ${\rm Aff}(V)$ forms the root system of type $A_{k-1}$. 
The Weyl group $W_{0}$ is generated by $\{s_{\alpha}\}_{\alpha \in R_{0}}$.  
We regard the set $R:=R_{0}+\mathbb{Z}(L\delta)$ as the affine root system of type $A_{k-1}^{(1)}$ 
with null roots $\mathbb{Z}(L\delta)$. 
The group $W$ generated by $\{s_{a}\}_{a \in R}$ is called 
the {\it affine Weyl group of type $A_{k-1}^{(1)}$}.  
Any element of $W$ is uniquely written in the form  
$wt_{L\beta}$ where $w \in W_{0}$ and 
$\beta$ is an element of the coroot lattice $Q^{\vee}:=\sum_{\alpha \in R_{0}}\mathbb{Z}\alpha^{\vee}$. 
In this sense we have $W=W_{0}\ltimes (LQ^{\vee})$. 
The gradient map $D \, : W \twoheadrightarrow W_{0}$ defined by 
$D(wt_{L\beta})=w$ is a group homomorphism.  

The {\it extended affine Weyl group} $\widehat{W}$ is generated by 
$\{s_{a}\}_{a \in R}$ and $\{t_{Lx}\}_{x \in X}$. 
Set 
\begin{align*}
\pi:=t_{Lv_{1}}s_{1} \cdots s_{k-1}.  
\end{align*}
Then $\widehat{W}$ is generated by $\pi$ and $W_{0}$. 

Set $a_{0}:=-\alpha_{1k}+L\delta$ and 
$a_{i}=\alpha_{i, i+1} \, (1\le i <k)$. 
The set $\{a_{0}, \ldots , a_{k-1}\}$ gives a basis of $R$. 
Denote $R^{\pm}$ for the set of the associated positive and negative roots. 

The length of $w \in W$ is defined by $\ell(w):=\#(R^{+}\cap w^{-1}R^{-})$. 
We abbreviate $s_{a_{i}}$ by $s_{i}$ $(0 \le i <k)$. 
If $w=s_{i_{1}}\cdots s_{i_{r}} \, (0 \le i_{1}, \ldots , i_{r}<k)$ is a reduced expression, 
then $r=\ell(w)$ and 
$R^{+}\cap w^{-1}R^{-}=\{s_{i_{r}}\cdots s_{i_{p+1}}(a_{i_{p}})\}_{p=1}^{r}$. 

For $v \in V$, set 
\begin{align*}
I(v):=\{ a \in R^{+} \, | \, a(v)<0\}.  
\end{align*}
For $0 \le i <k$ and $v \in V$, we have 
$\#I(s_{i}v)=\#I(v)-1$ if and only if $a_{i}(v)<0$,  
and then $I(s_{i}v)=s_{i}(I(v)\setminus\{a_{i}\})$. 
Note that $I(v)=\emptyset$ if and only if $v$ belongs to 
the closure of the fundamental chamber 
\begin{align*}
\overline{C_{+}}:=\{v \in V \, | \, a_{i}(v)\ge 0 \,\, (0\le \forall{i}<k)\}.  
\end{align*}
For any $v \in V$, the orbit $Wv$ intersects $\overline{C_{+}}$ at one point.  
Take a shortest element $w \in W$ such that $wv \in \overline{C_{+}}$. 
Then $I(v)=R^{+}\cap w^{-1}R^{-}$, and hence   
the shortest element is uniquely determined for each $v \in V$. 
Denote it by $w_{v}$. 

\begin{lem}\label{lem:key}
Suppose that $I(v_{1}) \subset I(v_{2})$. 
Then $w_{v_{2}}=w_{w_{v_{1}}v_{2}}w_{v_{1}}$ and $\ell(w_{v_{2}})=\ell(w_{w_{v_{1}}v_{2}})+\ell(w_{v_{1}})$. 
\end{lem}

\begin{proof}
Let $w_{v_{1}}=s_{i_{1}} \cdots s_{i_{r}}$ be a reduced expression. 
Then
\begin{align*}
I(s_{i_{p}}\cdots s_{i_{r}}v_{2})=s_{i_{p}}(I(s_{i_{p+1}}\cdots s_{i_{r}}v_{2})\setminus\{a_{i_{p}}\})  
\end{align*}
for $1 \le p \le r$ because $I(v_{1}) \subset I(v_{2})$. 
Therefore $I(v_{2})=I(v_{1}) \sqcup w_{v_{1}}^{-1}I(w_{v_{1}}v_{2})$ 
and $\ell(w_{v_{2}})=\ell(w_{w_{v_{1}}v_{2}})+\ell(w_{v_{1}})$.    
Since $w_{w_{v_{1}}v_{2}}w_{v_{1}}$ moves $v_{2}$ into $\overline{C_{+}}$, 
it is equal to $w_{v_{2}}$. 
\end{proof}


\section{Definition of Hamiltonian}\label{sec:hamiltonian}

Denote the $\mathbb{C}$-vector space of $\mathbb{C}$-valued functions on $X$ by $F(X)$ . 
For $1 \le i \le k$, define $d^{\pm}_{i} \in F(X)$ by 
\begin{align}
d^{+}_{i}(x)&:=\#\{1 \le p <k \, | \, \sum_{j=i}^{i+p-1}a_{j}(x) \in L\mathbb{Z}_{\le 0}\}, 
\label{eq:def-d1} \\ 
d^{-}_{i}(x)&:=\#\{1 \le p <k \, | \, \sum_{j=i-p}^{i-1}a_{j}(x) \in L\mathbb{Z}_{\le 0}\}, 
\label{eq:def-d2} 
\end{align}
where the index $j$ of simple root $a_{j}$ is read modulo $k$. 
These functions have the following property: 

\begin{prop}\label{prop:d-change}
For $x \in X, 1 \le i \le k $ and $0 \le j<k$, we have 
\begin{align*}
d_{i}^{\pm}(s_{j}x)=\left\{ 
\begin{array}{ll}
d_{i}^{\pm}(x) & (i\not=j, j+1),\\
d_{j+1}^{\pm}(x)\pm\theta(a_{j}(x)=0) & (i=j),\\
d_{j}^{\pm}(x)\mp\theta(a_{j}(x)=0) & (i=j+1),\\
\end{array}
\right. 
\end{align*}
where $\theta(P)=1$ or $0$ if $P$ is true or false, respectively. 
\end{prop}

Now we define the operator $H$ acting on $F(X)$ by 
\begin{align*}
H:=\sum_{i=1}^{k}\beta^{d_{i}^{-}}(t_{v_{i}}-\alpha d_{i}^{+}),  
\end{align*}
where $\alpha \in \mathbb{C}$ and $\beta \in \mathbb{C}^{\times}$ are constants. 
Setting $\alpha=-1$ and $\beta=1$, we have 
\begin{align*}
H=\sum_{i=1}^{k}t_{v_{i}}+\sum_{1 \le i<j \le k \atop m \in \mathbb{Z}}
\theta(\alpha_{ij}(\cdot)+mL=0).  
\end{align*}
It gives the discrete Hamiltonian $\tilde{H}$ \eqref{eq:OC-Y} 
with periodic boundary condition of size $L$ up to an additive constant.  

The operator $H$ is $W$-invariant in the following sense. 
Set 
\begin{align*}
X_{\rm reg}:=X-\bigcup_{a \in R^{+}}V_{a}. 
\end{align*}

\begin{prop}
For $f \in F(X)$ and $w \in W$, it holds that $wHw^{-1}f=Hf$ on $X_{\rm reg}$. 
\end{prop}

\begin{proof}
Take a point $x \in X_{\rm reg}$ and $w \in W$. 
Define $\mu \in \mathfrak{S}_{k}$ by $(Dw)(v_{i})=v_{\mu(i)} \, (1 \le i \le k)$, then 
Proposition \ref{prop:d-change} implies that 
$d_{i}^{\pm}(w^{-1}x)=d_{\mu(i)}^{\pm}(x)$. 
Hence we have 
\begin{align*}
(wHw^{-1}f)(x)=\sum_{i=1}^{k}\beta^{d_{\mu(i)}^{-}(x)}
( f(wt_{-v_{i}}w^{-1}x)-\alpha d_{\mu(i)}^{+}(x)f(x)  ).  
\end{align*}
Since $wt_{-v_{i}}w^{-1}=t_{-(Dw)(v_{i})}=t_{-v_{\mu(i)}}$, 
the right hand side is equal to $(Hf)(x)$. 
\end{proof}


\section{Integral-reflection operators}\label{sec:refl}

\subsection{Affine Hecke algebra}

\begin{definition}
The {\it affine Hecke algebra $\widehat{\mathcal{H}}$ of type $GL_{k}$} is 
the unital associative algebra with generators 
$T_{i} \, (1 \le i<k)$ and $Y_{i} \, (1 \le i \le k)$ satisfying
\begin{align}
& 
(T_{i}-1)(T_{i}+\beta)=0 \quad (1 \le i<k), 
\label{eq:quad} \\ 
& 
T_{i}T_{i+1}T_{i}=T_{i+1}T_{i}T_{i+1} \quad (1 \le i \le k-2), \quad 
T_{i}T_{j}=T_{j}T_{i} \quad (|i-j|>1), 
\label{eq:braid} \\ 
& 
Y_{i}Y_{j}=Y_{j}Y_{i} \quad (1 \le i, j \le k), 
\nonumber \\ 
& 
Y_{i}T_{j}=T_{j}Y_{i} \quad (j \not=i-1, i), \quad  
T_{i}Y_{i+1}T_{i}=Y_{i} \quad (1 \le i<k). 
\nonumber 
\end{align}
\end{definition}

Set $\omega:=Y_{k}T_{k-1}\cdots T_{1}$. Then $\omega T_{i}=T_{i-1} \omega \, (1<i<k)$ and 
$\omega^{2}T_{1}=T_{k-1}\omega^{2}$. 
The subalgebra $\mathcal{H}$ generated by $T_{i} \, (1 \le i<k)$ and $T_{0}:=\omega T_{1} \omega^{-1}$ 
is called the {\it affine Hecke algebra of type $A_{k-1}^{(1)}$}. 

\subsection{Integral-reflection operators}

Let us construct the integral-reflection operators acting on $F(X)$ from 
divided difference operators acting on polynomial ring following \cite{ED}. 

We identify the group algebra $\mathbb{C}[X]$ with 
the Laurent polynomial ring $\mathbb{C}[e^{\pm v_{1}}, \ldots , e^{\pm v_{k}}]$. 
The extended affine Weyl group acts on $\mathbb{C}[X]$ by 
$w(e^{x}):=e^{wx} \, (w \in \widehat{W}, x \in X)$.  
Consider the operator $\check{T}_{i} \, (1 \le i<k)$ acting on $\mathbb{C}[X]$ defined by 
\begin{align*}
\check{T}_{i}:=s_{i}+\frac{\alpha e^{v_{i+1}}+1-\beta}{1-e^{-v_{i}+v_{i+1}}}(1-s_{i}).  
\end{align*}
In \cite{LS} Lascoux and Sch\"utzenberger characterized 
divided difference operators acting on polynomial ring 
which satisfy the braid relation. 
They are parameterized by four parameters. 
The operators $\check{T}_{i} \, (1 \le i<k)$ are obtained 
by setting one of the parameters to zero, 
and they satisfy the quadratic relation \eqref{eq:quad} and the braid relation \eqref{eq:braid}. 

The group algebra $\mathbb{C}[X]$ acts on $F(X)$ by $e^{x}f:=t_{-x}f$ 
\footnote{This action is different from that in \cite{ED} where $e^{x}f:=t_{x}f$.}. 
Define a non-degenerate bilinear pairing $F(X) \times \mathbb{C}[X] \to \mathbb{C}$ by 
$(f, p):=(pf)(0)$.  
We consider the operator $Q_{i} \, (1 \le i<k)$ acting on $F(X)$ determined by 
$(Q_{i}f, p)=(f, \check{T}_{i}p) \, (\forall{p} \in \mathbb{C}[X])$.  
An explicit formula for $Q_{i}$ is given as follows. 

\begin{definition}
The {\it integral-reflection operator} $Q_{i} \, (1 \le i<k)$ acting on $F(X)$ is defined by 
\begin{align*}
(Q_{i}f)(x):=\left\{ 
\begin{array}{ll}
\displaystyle 
f(s_{i}x)+\sum_{j=1}^{a_{i}(x)}\left( \alpha f(s_{i}x+ja_{i}^{\vee}+v_{i+1})+
(1-\beta) f(s_{i}x+ja_{i}^{\vee})\right) & 
(a_{i}(x)>0), \\ 
f(x) & (a_{i}(x)=0), \\
\displaystyle 
f(s_{i}x)-\!\!\!\sum_{j=0}^{-a_{i}(x)-1}\!\!\! 
\left( \alpha f(s_{i}x-ja_{i}^{\vee}+v_{i+1})+(1-\beta) f(s_{i}x-ja_{i}^{\vee})\right) &  
(a_{i}(x)<0). 
\end{array}
\right.
\end{align*} 
\end{definition}

The operators $Q_{i} \, (1\le i<k)$ satisfy the same quadratic relations and braid ones as 
$\check{T}_{i} \, (1 \le i<k)$ because the relations are left-right symmetric. 
Denote by $\check{\pi}$ the action of $\pi \in \widehat{W}$ on $\mathbb{C}[X]$. 
It satisfies $\check{T}_{i}\check{\pi}=\check{\pi}\check{T}_{i-1} \, (1<i<k)$ and 
$\check{T}_{1}\check{\pi}^{2}=\check{\pi}^{2}\check{T}_{k-1}$. 
{}From $(\pi^{-1}f, p)=(f, \check{\pi}p)$,
we find that $Q_{i} \, (1 \le i<k)$ and $\pi$ give 
$\widehat{\mathcal{H}}$-module structure on $F(X)$: 

\begin{prop}\label{prop:H-module}
The assignment $T_{i} \mapsto Q_{i} \, (1 \le i <k)$ and $\omega \mapsto \pi^{-1}$ 
extends uniquely to a representation $\rho \, : \, \widehat{\mathcal{H}} \to {\rm End}_{\mathbb{C}}F(X)$ 
of the affine Hecke algebra of type $GL_{k}$. 
\end{prop}

In the rest of this paper we make use of the restriction of $\rho$ to 
the subalgebra $\mathcal{H}$. 


\section{Propagation operator}\label{sec:main}

Set $Q_{0}:=\rho(T_{0})=\pi^{-1}Q_{1}\pi$. 
Let $w=s_{i_{1}} \cdots s_{i_{m}}$ be a reduced expression of $w \in W$ 
and set $Q_{w}:=Q_{i_{1}} \cdots Q_{i_{m}}$. 
It does not depend on the choice of reduced expression of $w$.  

\begin{thm}\label{thm:main}
For $f \in F(X)$, define $G(f) \in F(X)$ by 
\begin{align*}
G(f)(x):=(w_{x}^{-1}Q_{w_{x}}f)(x).  
\end{align*}
If $f$ is an eigenfunction of the operator $\sum_{i=1}^{k}t_{\epsilon_{i}}$ 
with eigenvalue $\lambda \in \mathbb{C}$, 
then $G(f)$ satisfies $HG(f)=\lambda G(f)$. 
\end{thm}

The following lemma plays a key role in the proof of Theorem \ref{thm:main}. 
\begin{lem}\label{lem:main}
Let $f \in F(X)$ and $x \in X$. 
Define $\sigma \in \mathfrak{S}_{k}$ by $(Dw_{x})(v_{i})=v_{\sigma(i)} \, (1\le i \le k)$. 
Then we have 
\begin{align*}
((t_{v_{i}}-\alpha d_{i}^{+})G(f))(x)=
((t_{v_{\sigma(i)}}+(1-\beta)\sum_{j=1}^{d_{i}^{+}(x)}t_{v_{\sigma(i)+j}})Q_{w_{x}}f)(w_{x}x).   
\end{align*} 
for $1 \le i \le k$.  
In the right hand side the index $j$ of $v_{j}$ is read modulo $k$. 
\end{lem}

\begin{proof}
In the proof we fix $x$ and $i$, and set $y:=x-v_{i} \in X$,  
$x':=x-\frac{1}{2}v_{i} \in \mathbb{Q}\otimes_{\mathbb{Z}}X$, 
$l=d_{i}^{+}(x)$ and $p=\sigma(i)$.  
For any $a \in R^{+}$, it holds that $a(x')=a(x)-(Da)(v_{i})/2=a(y)+(Da)(v_{i})/2$. 
Since $|(Da)(v_{i})| \le 1$, the two sets $I(x)$ and $I(y)$ are included in $I(x')$. 
Therefore $w_{x'}=w_{w_{x}x'}w_{x}=w_{w_{y}x'}w_{y}$ and 
$\ell(w_{x'})=\ell(w_{w_{x}x'})+\ell(w_{x})=\ell(w_{w_{y}x'})+\ell(w_{y})$ from 
Lemma \ref{lem:key}. 

Let us write down $w_{w_{x}x'}$. 
Since $w_{x}x \in \overline{C_{+}}$ and $a(w_{x}x')=a(w_{x}x)-(Da)(v_{p})/2$ for any $a \in R$, 
we have
\begin{align*}
I(w_{x}x')=\{ a \in R^{+} \, | \, a(w_{x}x)=0, \, (Da)(v_{p})>0 \}.  
\end{align*}
Now note that $l=d_{p}^{+}(w_{x}x)$ because $w_{x}$ is shortest. 
If $z \in \overline{C_{+}}$, it holds that 
$d_{j}^{+}(z)={\rm max}\{0 \le c \le k-1 \, | \, \sum_{r=j}^{j+c-1}a_{r}(z)=0 \}$.  
Therefore 
$I(w_{x}x')=\{ s_{p} \cdots s_{p+j-1}(a_{p+j})\}_{j=0}^{l}$, 
where the index $j$ of $s_{j}$ and $a_{j}$ is read modulo $k$. 
Thus we get 
\begin{align}\label{eq:shortest1}
w_{w_{x}x'}=s_{p+l-1} \cdots s_{p+1} s_{p}.  
\end{align}
Note that $a_{p+j}(w_{x}x)=0$ for $0 \le j <l$. 

Starting from the fact 
\begin{align*}
I(w_{y}x')=\{ a \in R^{+} \, | \, a(w_{y}y)=0, \, (Da)(v_{q})<0 \},    
\end{align*}
where $(Dw_{y})(v_{i})=v_{q}$, 
we see that 
\begin{align}\label{eq:shortest2}
w_{w_{y}x'}=s_{q-l'}\cdots s_{q-1}, 
\end{align}
where $l':=d_{i}^{-}(y)$. 
The index $j$ of $s_{j}$ in the right hand side is read modulo $k$. 
Here note that $a_{q-j}(w_{y}y)=0$ for $1 \le j \le l'$. 

{}From \eqref{eq:shortest1} and \eqref{eq:shortest2}, we find that 
\begin{align*}
& 
s_{q-l'}\cdots s_{q-1}w_{y}=s_{p+l-1} \cdots s_{p+1} s_{p} w_{x}, \\ 
&  
Q_{w_{y}}=Q_{q-1}^{-1}\cdots Q_{q-l'}^{-1}
Q_{p+l-1} \cdots Q_{p+1} Q_{p} Q_{w_{x}}.  
\end{align*}
The first relation above implies that 
\begin{align*}
w_{y}y=(s_{p+l-1} \cdots s_{p+1} s_{p}w_{x})(x-v_{i})=
w_{x}x-v_{p+l}.   
\end{align*}
Therefore 
\begin{align*}
G(f)(y)=(Q_{w_{y}}f)(w_{y}y)=
(t_{v_{p+l}}Q_{p+l-1} \cdots Q_{p+1} Q_{p} Q_{w_{x}}f)
(w_{x}x).  
\end{align*}
Using the commutation relation 
\begin{align}\label{eq:tQ}
t_{v_{j+1}}Q_{j}=Q_{j}t_{v_{j}}+\alpha+(1-\beta)t_{v_{j+1}}, \quad 
t_{v_{j'}}Q_{j}=Q_{j}t_{v_{j'}} \, (j'\not=j, j+1),  
\end{align}
and $a_{p+j}(w_{x}x)=0 \, (0 \le j <l)$, we obtain the desired formula.  
\end{proof}

Now let us prove Theorem \ref{thm:main}. 
We fix $x \in X$, and let $\sigma$ be the permutation given in Lemma \ref{lem:main}.  
Identify the set $\{1, \ldots , k\}$ with $\mathbb{Z}/k\mathbb{Z}$, 
and decompose it into intervals of the form $[p, p+l] \, (1\le p \le k, \, 0\le l \le k-1)$ 
having the property 
\begin{align*}
a_{p-1}(w_{x}x)>0, \quad a_{p+j}(w_{x}x)=0 \,\, (0 \le j<l), \quad a_{p+l}(w_{x}x)>0.  
\end{align*}
Take one interval $[p, p+l]$. 
Then $d_{\sigma^{-1}(p+j)}^{+}(x)=d_{p+j}^{+}(w_{x}x)=l-j$ and 
$d_{\sigma^{-1}(p+j)}^{-}(x)=d_{p+j}^{-}(w_{x}x)=j$ for $0\le j \le l$. 
{}From Lemma \ref{lem:main} we have 
\begin{align*}
& 
\sum_{j=0}^{l}\beta^{d_{\sigma^{-1}(p+j)}^{-}}
((t_{v_{\sigma^{-1}(p+j)}}-\alpha d_{\sigma^{-1}(p+j)}^{+})G(f))(x) \\ 
&=\sum_{j=0}^{l}\beta^{j}
((t_{v_{p+j}}+(1-\beta)\sum_{r=1}^{l-j}t_{v_{p+j+r}})Q_{w_{x}}f)(w_{x}x)=\sum_{j=0}^{l}
(t_{v_{p+j}}Q_{w_{x}}f)(w_{x}x).    
\end{align*}
The above relation holds on each interval, and we get 
\begin{align*}
(HG(f))(x)=\sum_{i=1}^{k}(t_{v_{i}}Q_{w_{x}}f)(w_{x}x).  
\end{align*}
{}From \eqref{eq:tQ}, the operator $\sum_{i=1}^{k}t_{v_{i}}$ commutes with $Q_{w} \, (w \in W)$. 
Therefore if $f$ is an eigenfunction of $\sum_{i=1}^{k}t_{v_{i}}$ with eigenvalue $\lambda$, 
it holds that $(HG(f))(x)=\lambda f(x)$. 
This completes the proof.


\section{Bethe wave functions}\label{sec:bethe}

Let us construct $\widehat{W}$-invariant eigenfunctions by means of the Bethe ansatz method: 

\begin{prop}
Suppose that $p=(p_{1}, \ldots , p_{k}) \in (\mathbb{C}^{\times})^{k}$ is 
a solution of the system of algebraic equations
\begin{align}\label{eq:bethe-eq}
p^{L}_{i}=\prod_{j=1 \atop (j\not=i)}^{k}
\frac{\beta p_{i}-p_{j}-\alpha}{p_{i}-\beta p_{j}+\alpha}
\quad (1 \le i \le k).  
\end{align} 
Define the function $h_{p}$ by 
\begin{align}\label{eq:bethe-sol}
h_{p}(x)=\sum_{\sigma \in \mathfrak{S}_{k}}\mathop{\rm sgn}(\sigma)\!\!\! 
\prod_{1\le i<j \le k}(\beta p_{\sigma(i)}-p_{\sigma(j)}-\alpha)
\prod_{i=1}^{k}p_{\sigma(i)}^{-\epsilon_{i}(x)} 
\quad (x \in \overline{C_{+}}) 
\end{align}
and $h_{p}(w x)=h_{p}(x)$ for any $w \in W$.  
Then $h_{p}$ is $\widehat{W}$-invariant and 
an eigenfunction of $H$ with eigenvalue $\sum_{i=1}^{k}p_{i}$. 
\end{prop}

\begin{proof}
Denote by $f_{p} \in F(X)$ the function defined by 
the right hand side of \eqref{eq:bethe-sol} on the whole $X$. 
The function $g_{p}(x):=\prod_{i=1}^{k}p_{i}^{-\epsilon_{i}(x)}$ satisfies 
\begin{align*}
Q_{i}g_{p}=s_{i}g_{p}+
\frac{\alpha+(1-\beta) p_{i+1}}{p_{i}-p_{i+1}}(s_{i}-1)g_{p}.
\end{align*}
Hence $Q_{i}f_{p}=f_{p}$ for all $1 \le i <k$. 
Moreover it holds that $\pi f_{p}=f_{p}$ if $\{p_{i}\}_{i=1}^{k}$ is 
a solution to \eqref{eq:bethe-eq}. 
Therefore we get $G(f_{p})(x)=f_{p}(w_{x}x)=h_{p}(w_{x}x)=h_{p}(x)$ for any $x \in X$. 
{}From the equality $\sum_{i=1}^{k}t_{v_{i}}f_{p}=(\sum_{i=1}^{k}p_{i})f_{p}$ 
and Theorem \ref{thm:main}, 
we find that $h_{p}$ is an eigenfunction of $H$ with eigenvalue $\sum_{i=1}^{k}p_{i}$. 
Since $\widehat{W}x\cap\overline{C_{+}}=\{\pi^{n}w_{x}x\}_{n \in \mathbb{Z}}$, 
we have $G(f_{p})(\pi x)=G(f_{p})(x)$. 
Hence $h_{p}$ is $\widehat{W}$-invariant. 
\end{proof}

In the case of $\alpha=0$ the Bethe wave function $h_{p}(x)$ \eqref{eq:bethe-sol} 
can be written in terms of the Hall-Littlewood polynomials \cite{Mac}. 
Let $t \in \mathbb{C}^{\times}$ be a parameter. 
For a partition $\lambda=(\lambda_{1}, \ldots , \lambda_{k})$, 
we set 
\begin{align*}
v_{\lambda}(t)=\prod_{a \ge 1}\prod_{n=1}^{m_{a}}\frac{1-t^{n}}{1-t}, 
\end{align*}
where $m_{a}$ is the number of $\lambda_{j}$ equal to $a$. 
The Hall-Littlewood polynomial $P_{\lambda}$ is defined by 
\begin{align*}
P_{\lambda}(z_{1}, \ldots , z_{k} ; t):=\frac{1}{v_{\lambda}(t)}
R_{\lambda}(z_{1}, \ldots , z_{k} ; t), 
\end{align*}
where 
\begin{align*}
R_{\lambda}(z_{1}, \ldots , z_{k} ; t):=\sum_{\sigma \in \mathfrak{S}_{k}}
\prod_{1 \le i<j \le k}\frac{z_{\sigma(i)}-tz_{\sigma(j)}}{z_{\sigma(i)}-z_{\sigma(j)}}
\prod_{i=1}^{k}z_{\sigma(i)}^{\lambda_{i}}. 
\end{align*}
Then the Bethe wave function with $\alpha=0$ is written as 
\begin{align}
h_{p}(x)|_{\alpha=0}=\Delta(p)R_{\epsilon(x)}(p_{1}^{-1}, \ldots , p_{k}^{-1} ; \beta) \quad 
(x \in \overline{C_{+}}),   
\label{eq:HL}
\end{align}
where $\Delta(p):=\prod_{1\le i<j \le k}(p_{i}-p_{j})$ and 
$\epsilon(x):=(\epsilon_{1}(x), \ldots , \epsilon_{k}(x))$.


\medskip 
\section*{Acknowledgments}

The research of the author is supported by Grant-in-Aid for 
Young Scientists (B) No.\,23740119. 
The author is grateful to Saburo Kakei and Tomohiro Sasamoto for discussions. 
He also thanks the referee for valuable comments. 


\end{document}